\documentclass[twoside,leqno]{article}

\usepackage{ltexpprt}
\usepackage{hyperref}

\usepackage{floatrow}
\usepackage{amsmath}
\usepackage{amssymb}
\usepackage{amsfonts}
\usepackage{color}
\usepackage{algorithm}
\usepackage{caption}
\usepackage{MnSymbol}
\usepackage{url}
\usepackage[margin=1in]{geometry}
\usepackage[labelfont=bf,font=small,margin=1em]{caption} 
\usepackage{mathtools}
\usepackage{refcount}
\usepackage[shortlabels]{enumitem}

\newcommand\ce{\coloneq}

\newlength\tmp
\newcommand\twopart[2]{%
	\settowidth\tmp{#1}%
	#1\mbox{\parbox[t]{\linewidth - \tmp}{\raggedright{}#2}}%
}

\begin{document}

\newcommand\relatedversion{}

\title{\Large Selectable Heaps and Optimal Lazy Search Trees}
\author{Bryce Sandlund\thanks{David R. Cheriton School of Computer Science, University of Waterloo}
\and Lingyi Zhang\thanks{David R. Cheriton School of Computer Science, University of Waterloo}}

\date{}

\maketitle


\fancyfoot[R]{\scriptsize{Copyright \textcopyright\ 2022 by SIAM\\
Unauthorized reproduction of this article is prohibited}}





\begin{abstract} \small\baselineskip=9pt We show the $O(\log n)$ time extract minimum function of efficient priority queues can be generalized to the extraction of the $k$ smallest elements in $O(k \log(n/k))$ time\footnote{We define $\log(x)$ as $\max(\log_2(x), 1)$.}, which we prove optimal for comparison-based priority queues with $o(\log n)$ time insertion. We show heap-ordered tree selection (Kaplan et al., SOSA '19) can be applied on the heap-ordered trees of the classic Fibonacci heap and Brodal queue, in $O(k \log(n/k))$ amortized and worst-case time, respectively. We additionally show the deletion of $k$ elements or selection without extraction can be performed on both heaps, also in $O(k \log(n/k))$ time. Surprisingly, all operations are possible with no modifications to the original Fibonacci heap and Brodal queue data structures.

We then apply the result to lazy search trees (Sandlund \& Wild, FOCS '20), creating a new interval data structure based on selectable heaps. This gives optimal $O(B+n)$ time lazy search tree performance, lowering insertion complexity into a gap $\Delta_i$ from $O(\log(n/|\Delta_i|) + \log \log n)$ to $O(\log(n/|\Delta_i|))$ time. An $O(1)$ time merge operation is also made possible when used as a priority queue, among other situations. If Brodal queues are used, all runtimes of the lazy search tree can be made worst-case.\end{abstract}

\section{Introduction}
\label{sec:intro}
The first priority queue data structure, binary heaps, were invented in 1964 for the heapsort algorithm~\cite{Williams64}. Binary heaps support $O(\log n)$ time extract minimum and insert operations. Due to their simplicity and storage of elements in an array, binary heaps or their generalization to $d$-ary heaps~\cite{Johnson75,Tarjan83} continue to be one of the most practical priority queues. However, some priority queue applications require additional operations. The binomial heap was invented in 1978~\cite{Vuillemin78}, supporting insertion in $O(1)$ amortized time and the merge of two heaps also in $O(1)$ amortized time~\cite{Khoong93}, allowing for more efficient minimum spanning tree algorithms~\cite{Tarjan83}.

A breakthrough in efficient priority queue research came in 1984 with the development of Fibonacci heaps~\cite{Fredman87}. Fibonacci heaps generalize binomial heaps to support an efficient decrease-key operation, allowing for a faster implementation of Dijkstra's single-source shortest path algorithm~\cite{Dijkstra59}, among other applications~\cite{Fredman87}. Fibonacci heaps can perform insert, merge, and decrease-key all in $O(1)$ amortized time.
A number of priority queues with time bounds matching or close to Fibonacci heaps have since been developed~\cite{Fredman86,Chan09,Brodal12,Elmasry09,Haeupler11,Brodal96,Hansen15,Driscoll88,Hoyer95,Brodal95,Elmasry10,Kaplan99,Kaplan08,Sandlund20}. Many claim to be a simpler or more practical alternative to Fibonacci heaps~\cite{Chan09,Hansen15,Haeupler11,Elmasry10,Kaplan99,Kaplan08,Sandlund20}; however, Fibonacci heaps continue to be one of the most-taught, most-performant, and simplest-to-code priority queues with optimal theoretical efficiency~\cite{Hansen15,Larkin14}.

Brodal gave two priority queues matching Fibonacci heap bounds but in the worst-case. Brodal queues~\cite{Brodal96} were first presented in 1996, and the improved version strict Fibonacci heaps~\cite{Brodal12} (coauthored with Lagogiannis and Tarjan) were given in 2012. While Brodal queues make essential use of $O(1)$ time array access, strict Fibonacci heaps operate entirely on the pointer machine. More on the history and breadth of research on priority queues is given in the recent survey by Brodal~\cite{Brodal09}.

In this paper, we show the extract minimum function of priority queues, particularly, Fibonacci heaps~\cite{Fredman87} and Brodal queues~\cite{Brodal96}, can be generalized to the extraction of the $k$ smallest elements. We also show selection without extraction and deletion of multiple elements can be performed efficiently. Altogether, we consider the following set of heap operations, where we consider elements themselves to be keys (the generalization to key-value pairs is straightforward).
\begin{itemize}
	\item \texttt{MakeHeap()} $\ce$ Create a new, empty heap.
	\item \texttt{Merge($h_1$,\,$h_2$)} $\ce$ Return a new heap containing all elements of $h_1$ and $h_2$, destroying $h_1$ and $h_2$.
	\item \texttt{Insert($e$)} $\ce$ Add element $e$ to the heap.
	\item \texttt{DecreaseKey($e$,\,$v$)} $\ce$ Decrease the value of $e$ to $v$, with a pointer to $e$.
	\item \texttt{Delete($e$)} $\ce$ Delete $e$ from the heap, with a pointer to $e$.
	\item \texttt{FindMin()} $\ce$ Return the minimum element of the heap.
	\item \texttt{ExtractMin()} $\ce$ Return and remove the minimum element of the heap.
	\item \texttt{SelectK($k$)} $\ce$ Return the $k$ smallest elements of the heap, in no particular order.
	\item \texttt{ExtractK($k$)} $\ce$ Return and remove the $k$ smallest elements of the heap, in no particular order.
	\item \texttt{Delete($e_1, \ldots, e_k$)} $\ce$ Remove elements $e_1, \ldots, e_k$ from the heap, with a pointer to each.
\end{itemize}

\subsection{Results}
\sloppy
Recall that Fibonacci heaps support \texttt{MakeHeap()}, \texttt{Merge($h_1$,\,$h_2$)}, \texttt{Insert($e$)}, \texttt{DecreaseKey($e$,\,$v$)}, and \texttt{FindMin()} in $O(1)$ amortized time and \texttt{ExtractMin()} and \texttt{Delete($e$)} in $O(\log n)$ amortized time~\cite{Fredman87}. Via repeated application of \texttt{ExtractMin()} and \texttt{Delete($e$)}, Fibonacci heaps also support \texttt{ExtractK($k$)} and \texttt{Delete($e_1, \ldots, e_k$)} in $O(k \log n)$ time.
We show \texttt{FindMin()} and \texttt{ExtractMin()} operations can be more directly generalized to \texttt{SelectK($k$)} and \texttt{ExtractK($k$)}, respectively. We apply the heap-ordered tree selection algorithm of Kaplan, Kozma, Zamir, and Zwick~\cite{Kaplan18} to the heap-ordered trees of Fibonacci heaps, supporting \texttt{SelectK($k$)} and \texttt{ExtractK($k$)} in $O(k \log(n/k))$ amortized time. No modifications to the Fibonacci heap data structure itself are necessary.

We also show that for a comparison-based priority queue, if \texttt{Insert($e$)} is supported in $o(\log n)$ time then \texttt{ExtractK($k$)} must take $\Omega(k \log(n/k))$ time for some (smaller) values of $k$ and further that if \texttt{Insert($e$)} is supported in $O(1)$ time then \texttt{ExtractK($k$)} must take $\Omega(k \log(n/k))$ time for all values of $k$. This lower bound is not surprising since we need to pay for sorting the remaining elements after extracting the $k$ smallest, but we save the cost for sorting the $k$ elements removed from the set.

We then show Brodal queues~\cite{Brodal96} also support efficient selection without modification. Brodal queues match all time bounds of Fibonacci heaps for the considered set of standard priority queue operations but in the worst-case. We show heap-ordered tree selection~\cite{Kaplan18} can be made to work on a Brodal queue, giving \texttt{SelectK($k$)} and \texttt{ExtractK($k$)} in $O(k \log(n/k))$ worst-case time.
This provides a priority queue with optimal worst-case time operations for all standard operations while also supporting optimal worst-case time extraction of the $k$ smallest elements. Finally, we show both Fibonacci heaps and Brodal queues support \texttt{Delete($e_1, \ldots, e_k$)} in $O(k \log(n/k))$ amortized and worst-case time, respectively.

The exact conditions necessary for efficient selection in a heap are delicate. In order to have the selection algorithm \cite{Kaplan18} fast enough for $k$ selection in a heap-ordered tree, we want to limit the number of nodes with large degrees. This is because whenever a node is selected, we must consider all its children as next possible smallest nodes. Therefore we wish that the total degree of selected nodes be minimized.

One of the necessary conditions will be that the degree of each node must be bounded by the logarithm of its subtree size. Intuitively, if this property is satisfied and the descendants of selected nodes disjoint, we can apply Jensen's inequality to show the total degree of selected nodes is $O(k \log (n/k))$. Unfortunately, we do not have the guarantee that descendants of selected nodes will be disjoint. For example, a chain of nodes each with $O(\log n)$ child leaves could give a total degree sum of $\Omega(k \log n)$ if all selected nodes are along the chain.

Fibonacci heaps do not permit such structures because child subtrees must be of exponentially-increasing size, the key property we use in our proof. Brodal queues follow even more-rigorous structure constraints, and the technical difficulty of our proof for efficient selection is getting around the potentially $O(n)$ nodes throughout the heap that do not satisfy heap order. A careful analysis of invariants O4 and O5 of the Brodal queue structure can be used to show the number of such ``violating nodes" encountered during selection cannot be too many.

 Various existing priority queues seem roughly evenly-divided in their ability to also support efficient selection.
Specifically, it appears relaxed heaps~\cite{Driscoll88}, two-tier pruned binomial queues~\cite{Elmasry07}, hollow heaps~\cite{Hansen15}, and thin and fat heaps~\cite{Kaplan99} can also support selection in $O(k \log (n/k))$ time, while quake heaps~\cite{Chan09}, strict Fibonacci heaps~\cite{Brodal12}, pairing heaps~\cite{Fredman86}, rank-pairing heaps~\cite{Haeupler11}, and violation heaps~\cite{Elmasry10} do not.

It may seem surprising that existing efficient priority queues can support the selection of the $k$ smallest elements in optimal time without modification, despite being designed only for the extraction of the minimum element. The fact these classic heap data structures can support this behavior suggests they make few comparisons beyond the minimal required of the underlying partial order and keep this information efficiently accessible.

\subsection{Application to lazy search trees}

We develop selectable heaps for the application of an interval data structure for use in a lazy search tree~\cite{Sandlund20}. A lazy search tree is a comparison-based data structure that supports the operations of a binary search tree; this operation set is referred to as a \emph{sorted dictionary} in~\cite{Sandlund20}. (A formal description of allowed operations is given in Section~\ref{sec:lst}.) Instead of sorting elements on insertion, as does a binary search tree, lazy search trees store bags of unsorted elements in a partition into \emph{gaps} based on key order. Specifically, a set of gaps $\Delta_1, \ldots, \Delta_m$ are maintained such that for any $x \in \Delta_i$ and $y \in \Delta_{i+1}$, $x \leq y$. Inserted elements are placed into a gap respecting the key-order partition, and each query falls into a gap and \emph{splits} the gap into two new gaps at a position associated with the query operation.

Lazy search trees are able to provide superior runtimes to binary search trees on operation sequences with few queries or non-uniform query distribution. For example, $n$ insertions and $q$ queries can be served in $O(n \log q + q \log n)$ time (optimal), $q$ queries for $k$ consecutive keys with $n$ interspersed uniformly-distributed insertions can be served in $O(n \log q + qk \log n + n \log \log n)$ time (additive $O(n \log \log n)$ time from optimal), or the data structure can be used directly as a priority queue with $O(\log \log n)$ time insert and decrease-key operations ($O(1)$ time insertion and decrease-key is optimal)~\cite{Sandlund20}. More generally, if we take $B = \sum_{i=1}^m |\Delta_i| \log_2(n/|\Delta_i|)$, where $|\Delta_i|$ denotes the size of gap $\Delta_i$ at the end of operations, lazy search trees serve an operation sequence of $n$ insertions and $q$ distinct queries in $O(B + \min(n \log \log n, n \log q))$ time, where $\Omega(B+n)$ is a lower bound. Per-operation performance statements and additional applications are stated in Theorem 1 of~\cite{Sandlund20} and Section 1.2 of~\cite{Sandlund20}, respectively. Lazy search trees give a new model which makes direct comparison with previous work more challenging, therefore we refer the interested reader to consult Section 2 of~\cite{Sandlund20} for historical context.

The interval data structure given in the original lazy search tree paper can be made into a selectable heap with time complexities matching those stated herein for Fibonacci heaps and Brodal queues, except that \texttt{Insert()} and \texttt{DecreaseKey()} take $O(\log \log n)$ time and \texttt{SelectK($k$)} can be supported in $O(k)$ amortized time. (Although operations \texttt{FindMin()} and \texttt{SelectK($k$)} are not explicitly addressed in~\cite{Sandlund20}, in Section 1.2, ``Selectable Priority Queue", it is stated how they can be performed.) By building an interval data structure off Fibonacci heaps or Brodal queues, we show the following:
\begin{enumerate}
	\item Lazy search trees can achieve optimal $O(B+n)$ time performance over a sequence of $n$ insertions and $q$ distinct queries, reducing insertion complexity into gap $\Delta_i$ from $O(\min(\log(n/|\Delta_i|) + \log \log |\Delta_i|, \; \log q))$ to $O(\log(n/|\Delta_i|))$. This reduces the time complexity for $n$ uniformly-distributed insertions and $q$ interspersed queries for $k$ consecutive keys from $O(n \log q + qk \log n + n \log \log n)$ to $O(n \log q + qk \log n)$ and improves insertion and decrease-key complexity as a priority queue from $O(\log \log n)$ to $O(1)$ time. (This answers open problem 2 from~\cite{Sandlund20}.)
	\item Lazy search trees can be made to support $O(1)$ time merge (worst-case time using Brodal queues, amortized time using Fibonacci heaps)
	when used as a priority queue (supported operations are listed in Section~\ref{sec:intro}), among other situations. (This answers open problem 4 from~\cite{Sandlund20}.)
	\item Queries in a lazy search tree can be made worst-case in the general case of two-sided gaps\footnote{A \emph{two-sided gap} is a gap $\Delta_i$ such that queries have occurred for ranks on both the left and right boundary of $\Delta_i$.}. (This addresses open problem 5 from~\cite{Sandlund20}; a fully-general worst-case time solution does not appear possible while keeping change-key in the exact model given in~\cite{Sandlund20}. We do offer a worst-case time solution with fully-general gap merge and change-key supported as deletion and re-insertion in $O(\log n)$ time.)
\end{enumerate}

The proposed data structure makes fundamental use of soft heaps~\cite{Chazelle00,Kaplan13,Kaplan18,Brodal21}, biased search trees, and efficient priority queues. In some sense, this research approaches a unification of different ordered data structures into a single theory of a best data structure for ordered data. Optimal priority queues (a single gap $\Delta_1$), online dynamic multiple selection~\cite{BARBAY16} (gaps are separated by queried ranks), and binary search trees (every element is in its own gap) are special cases of our solution. This paper closes the book theoretically in the gap-based model proposed in~\cite{Sandlund20}; any further work in this research direction requires generalization of the gap model, stated in~\cite{Sandlund20} as open problem 1.

\subsection{Organization}

We organize the remainder of this paper as follows. In Section~\ref{sec:lb}, we show a priority queue supporting insertion in $o(\log n)$ time must take $\Omega(k \log (n/k))$ time to extract the $k$ smallest elements.
In Section~\ref{sec:select}, we describe the general framework of heap-ordered tree selection from~\cite{Kaplan18}. In Section~\ref{sec:fib}, we show how to support \texttt{SelectK($k$)}, \texttt{ExtractK($k$)}, and \texttt{Delete($e_1, \ldots, e_k$)} in a Fibonacci heap in $O(k \log (n/k))$ amortized time. In Section~\ref{sec:brodal}, we show how to support \texttt{SelectK($k$)}, \texttt{ExtractK($k$)}, and \texttt{Delete($e_1, \ldots, e_k$)} in a Brodal queue in $O(k \log (n/k))$ worst-case time. In Section~\ref{sec:lst}, we apply selectable heaps to lazy search trees, achieving optimal performance, supporting merge as a priority queue, and giving worst-case time operations for the general case of two-sided gaps. We give concluding remarks in Section~\ref{sec:conclude}.

\subsection{Corequisite reading}
This paper relies on the internals of Fibonacci heaps~\cite{Fredman87} and Brodal queues~\cite{Brodal96}. In the interest of being self-contained, we have given overviews of relevant information in this text; however, familiarity with these data structures is useful. Further, this paper relies on the model and background given in lazy search trees~\cite{Sandlund20}. While the entire interval data structure and analysis can be skipped (we replace it here with selectable heaps), familiarity with sections 1 and 2 of~\cite{Sandlund20} will help put the given improvement into context. Finally, this paper relies on the heap-ordered tree selection algorithm of Kaplan, Kozma, Zamir, and Zwick~\cite{Kaplan18}, however we give an overview of the algorithm here and further the algorithm can be treated as a black box.

\section{Lower bound}
\label{sec:lb}

In this section, we give a simple comparison-based lower bound for extraction of the $k$ smallest elements in a priority queue.

\begin{theorem}
\label{thm:lb}
Fix $k$. A priority queue supporting insertion and extraction of the $k$ smallest elements must perform $\Omega(\log(n/k))$ comparisons on insertion or $\Omega(k \log(n/k))$ comparisons on extraction of the $k$ smallest elements.
\end{theorem}
\begin{proof}
We can reduce a multiple selection instance~\cite{Dobkin81} to the selectable priority queue. Consider the multiple selection instance in which we are to find the smallest, the $(k+1)$th smallest, $(2k+1)$th smallest, and so on, elements amongst a set of $n$ elements. The lower bound for this multiple selection instance is $\Omega(n\log(n/k))$ comparisons~\cite{KMMS}. We can reduce this to a heap supporting insertion and extraction of the $k$ smallest elements by first inserting all $n$ elements and then repeatedly extracting the $k$ smallest. The smallest amongst each set of $k$ can be determined with a simple scan. This implies either insertion must take $\Omega(\log(n/k))$ comparisons or extraction of the $k$ smallest elements must take $\Omega(k \log(n/k))$ comparisons.
\end{proof}

\begin{corollary}
\label{cor:lb}
A priority queue supporting \texttt{Insert($e$)} in $f(n)$ comparisons must spend $\Omega(k \log(n/k))$ comparisons on \texttt{ExtractK($k$)} operations for all $k = o(n/2^{f(n)})$.
\end{corollary}
\begin{proof}
Suppose insertion is supported in $f(n)$ time, which notably must be independent of $k$. Solve $\Omega(\log(n/x)) = f(n)$ for $x$. Then by Theorem~\ref{thm:lb}, \texttt{ExtractK($k$)} must take $\Omega(k \log(n/k))$ time for $k = o(x)$.
\end{proof}

Corollary~\ref{cor:lb} implies comparison-based priority queues with $o(\log n)$ time insertion must spend $\Omega(k \log(n/k))$ time on \texttt{ExtractK($k$)} operations for some (smaller) values of $k$. Further, priority queues with $O(1)$ time insertion must spend $\Omega(k \log(n/k))$ time on \texttt{ExtractK($k$)} operations for all values of $k$ (when $k = O(n)$, the size of output is $O(n)$ and so we must take $\Omega(n)$ time, and when $k = o(n)$ Corollary~\ref{cor:lb} implies \texttt{ExtractK($k$)} must take $\Omega(k \log(n/k))$ time).

\section{Heap-ordered tree selection}
\label{sec:select}

We use the heap-ordered tree selection algorithm \texttt{Soft-Select-Heapify($r$)} from Kaplan, Kozma, Zamir, and Zwick~\cite{Kaplan18}, where $r$ is the root of the tree. Algorithm \texttt{Soft-Select-Heapify($r$)} works on arbitrary heap-ordered trees in which nodes may have different degrees. We will use \texttt{Soft-Select-Heapify($r$)} as a black box. Given a heap-ordered tree with root $r$,  \texttt{Soft-Select-Heapify($r$)} will return a set of $k$ smallest items from the tree.
This is in contrast to a different algorithm by Frederickson~\cite{Frederickson93} which can be used to achieve the same result but is quite complicated.

We need the following definition from~\cite{Kaplan18}.

\begin{Definition}[\cite{Kaplan18}]
	
For any subtree, $S$ rooted at the root of $T$, let $\Delta(S)$ denote the sum of the degrees of the nodes of $S$ in the tree $T$. Then, define
\[
D(T,k):=\max(\Delta(S)) \text{ for all subtrees $S$ of size $k$ rooted at the root of $T$.}
\]
\end{Definition}

Kaplan et al.~\cite{Kaplan18} give the following examples. A large enough $d$-ary tree has $D(T, k) = dk$, as each node has degree $d$. As another example, consider a tree $T$ where each node at depth $i$ has degree $i+2$; then, $D(T,k) = \sum_{i=0}^{k-1} (i+2) = k(k+3)/2$, where the subtrees achieving this maximum are paths starting from the root. Their heap-ordered tree selection theorem is the following.

\begin{theorem}[\cite{Kaplan18}]
	\label{thm:selection}
	Let $T$ be a heap-ordered tree with root $r$. Algorithm \texttt{Soft-Select-Heapify($r$)} selects the set of $k$ smallest items in $T$ in $O(D(T,3k))$ time.
\end{theorem}

We give the high level idea of the algorithm \texttt{Soft-Select-Heapify($r$)} here. The algorithm  uses the soft heap~\cite{Chazelle00,Kaplan13,Kaplan18,Brodal21} for selection.  A soft heap is a variant on the simple heap data structure that allows ``corruption". A node is corrupted if the key is increased. The number of corrupted nodes in the tree can be controlled by a selected parameter $0 < \epsilon \leq 1/2$. More precisely, the guarantee offered by soft heaps is the following: for a fixed value $\epsilon$, at any point in time there will be at most $\epsilon n$ corrupted keys in the heap, where $n$ is the number of elements inserted across the lifetime of the heap (which can be more than the number of current elements).
Soft heaps can achieve constant amortized time \texttt{Delete}, \texttt{Insert}, \texttt{Merge}, \texttt{ExtractMin}, and \texttt{FindMin} operations. A deterministic linear time selection algorithm is one of the applications of soft heaps. 

Algorithm \texttt{Soft-Select-Heapify($r$)} starts by initializing an empty soft heap $Q$ and set $S$. The algorithm first adds the root $r$ of the input heap into the soft heap $Q$. The algorithm then performs $k-1$ iterations as follows. In each iteration, the algorithm first performs $(e,C):=\ $Q.\texttt{ExtractMin()}, where $e$ is the node with minimum key in $Q$ and $C$ is the set of nodes that are corrupted by the operation. If $e$ is not corrupted, it is added to $C$. Then the algorithm inserts all the children of nodes in $C$ into both $Q$ and $S$. After $k-1$ iterations, the algorithm performs a $k$ selection on set $S$ to find the $k$ smallest items required.

With a carefully chosen parameter $\epsilon=1/6$~\cite{Kaplan18}, we can limit the total number of corruptions of the above algorithm to $2k$. Thus, there will be at most $3k$ nodes whose children are inserted into the soft heap $Q$. We have the runtime of the algorithm to be  $O(D(T,3k))$.

\section{Fibonacci heap selection}
\label{sec:fib}

In this section we describe our algorithm and analysis for Fibonacci heap~\cite{Fredman87} selection. 
We will show how we apply the heap-ordered tree selection algorithm of Kaplan, Kozma, Zamir, and Zwick~\cite{Kaplan18} to expand \texttt{FindMin()} and \texttt{ExtractMin()} operations to \texttt{SelectK($k$)} and \texttt{ExtractK($k$)} and prove that \texttt{SelectK($k$)} and \texttt{ExtractK($k$)} operations can be done in $O(k \log(n/k))$ amortized time. 

Fibonacci heaps store a forest of heap-ordered trees and a pointer to the minimum root. \texttt{FindMin()} returns the minimum root that the pointer points to. \texttt{Insert()} adds a single new root with no children to the forest, possibly updating the pointer to the minimum root. \texttt{Merge()} combines two forests into one, again possibly updating the minimum pointer. \texttt{ExtractMin()} removes the smallest element, makes all its children new roots in the forest, finds the next smallest element, and then repeatedly combines roots of the same degree. \texttt{DecreaseKey()} is supported with a marking scheme. Each non-root node can either be marked or unmarked. If a node $x$ is marked, this implies $x$ has lost a child. When a marked node $x$ loses a second child, $x$ is removed from the list of children stored at its parent, $p(x)$, the subtree rooted at $x$ is made a new root, and $x$ is unmarked. The parent $p(x)$ is then either marked or also removed to form a new root, recursively. \texttt{DecreaseKey()} of a node $x$ simply makes $x$ a new unmarked root and processes its parent $p(x)$ via the marking scheme. \texttt{Delete()} calls \texttt{DecreaseKey()} on the element and decreases the key to $-\infty$. For more information on Fibonacci heaps, consult the original paper~\cite{Fredman87} or Wikipedia\footnote{\url{https://en.wikipedia.org/wiki/Fibonacci_heap}.}.

It is perhaps surprising that a tree $T$ of a Fibonacci heap satisfies $D(T,k) = O(k \log(n/k))$. The bound is easier to show on a binomial heap~\cite{Vuillemin78}, since a binomial tree of degree $k$ has exactly $2^{k}$ nodes. Fibonacci heaps have a more-flexible structure, in particular giving only \emph{lower bounds} on the degree of child nodes. We cannot find a fixed relation between the degree of the node and the size of the subtree rooted at such node. In order to bound the runtime of the selection algorithm, it will require bounding the total degree sum by showing that a large degree sum implies more than $n$ nodes in the Fibonacci heap. The fundamental property used is that children of a node have exponentially-increasing subtree size. In Fibonacci heaps, this is due to the following lemma.

\begin{lemma}[\cite{Fredman87}]
	\label{lem:child}
	Let $x$ be any node in a Fibonacci heap. Arrange the children of $x$ in the order they were linked to $x$, from earliest to latest. Then the $i$th child of $x$ has a degree of at least $i-2$.
\end{lemma}

Lemma~\ref{lem:child} can be used to prove the following corollary.

\begin{corollary}[\cite{Fredman87}]
	\label{cor:ssize}
	A node of degree $k$ in a Fibonacci heap has at least $F_{k+2} \geq \phi^k$ descendants, including itself, where $F_k$ is the $k$th Fibonacci number and $\phi = (1 + \sqrt{5})/2$ is the golden ratio.
\end{corollary}

We give the main theorem of the section below.

\begin{theorem}
\label{thm:fibselect}
\sloppy
	Fibonacci heaps support \texttt{SelectK($k$)}, \texttt{ExtractK($k$)}, and \texttt{Delete($e_1, \ldots, e_k$)} in $O(k \log (n/k))$ amortized time.
\end{theorem}
\begin{proof}
We first create a heap-ordered tree $T_{big}$ from the collection of roots stored in the Fibonacci heap by creating a dummy node $d$ with value $-\infty$ and linking all roots below it. We then perform \texttt{Soft-Select-Heapify($d$)}~\cite{Kaplan18} to select the $k+1$ smallest elements from $T_{big}$. We show $D(T_{big},k+1) = O(t + k \log(n/k))$, where $t$ is the number of roots the Fibonacci heap contains at the time of selection.

Consider the subtree $T_{select}$ of $T_{big}$ of selected nodes. Subtree $T_{select}$ contributes $O(k)$ to the degree of the selected nodes in tree $T_{big}$. Let us thus ignore this contribution and consider only unselected children of selected nodes. As there are a total $n+1$ nodes in $T_{big}$, we will maximize the sum of unselected children by attaching subtrees of smallest size to every selected node other than $d$, so that we may attach as many of them as possible. By Lemma~\ref{lem:child} and Corollary~\ref{cor:ssize}, the first two unselected subtrees we attach must be of degree at least $0$, containing $1$ element, then the third must be of degree $1$, containing $2$ elements, and the $j$th subtree must be of degree at least $j-2$, containing at least $\phi^{j-2}$ elements. The number of attached subtrees, ignoring the contribution of $T_{select}$, is maximized when we attach the same number $i$ to each selected node. Solving for $i$ in
\[
k + k \sum_{j=2}^i \phi^{j-2} \geq n-k,
\]
we can determine
\[
i \leq \frac{\log(1 + \frac{(\phi-1)(n-2k)}{k})}{\log \phi} + 1.
\]
It follows that $i \leq \log_\phi (n/k) + 1$. Adding back the contribution of $T_{select}$ affects the above analysis by no more than $k$, so that in total each selected node has degree $O(\log(n/k))$, besides dummy node $d$ which has degree $t$. Structural limitations imposed by the particular tree $T_{big}$ can only make average degree decrease by replacing smaller subtrees with larger subtrees. It thus follows that $D(T_{big},k) = O(t + k \log(n/k))$ and by Theorem~\ref{thm:selection}, selection in $T_{big}$ takes $O(t + k \log(n/k))$ time.

To complete operation \texttt{SelectK($k$)}, we reduce $t$ to at most $\log n$ by repeatedly combining roots of equal degree. By the potential function of the Fibonacci heap~\cite{Fredman87}, this releases $t$ units of potential and achieves amortized $O(k \log(n/k))$ time selection. To complete operation \texttt{ExtractK($k$)}, we first remove the $k$ smallest elements. This leaves $O(t + k \log(n/k))$ independent trees. We then again repeatedly combine roots of equal degree, resulting in at most $\log n$ independent trees. The amortized time complexity of \texttt{ExtractK($k$)} is thus also $O(k \log(n/k))$.

The above degree bounds did not require the subtree $T_{select}$ be minimal, or even connected. Our argument shows the total sum of degrees of any $k$ nodes in a Fibonacci heap is $O(k \log(n/k))$. This allows an $O(k \log(n/k))$ time \texttt{Delete($e_1, \ldots, e_k$)} operation. For each $i$, we remove the subtree rooted at $e_i$ and perform cascading cuts until an ancestor of $e_i$ is not marked or a root is reached, as in the decrease-key operation. The final unmarked ancestor, if not a root, is then marked and all children of $e_i$ are added to the collection of roots. The total number of children is $O(k \log(n/k))$ and each cascading cut operation takes $O(1)$ amortized time per $e_i$. Thus the entire deletion can be performed in $O(k \log(n/k))$ amortized time.
\end{proof}

\section{Brodal queue selection}
\label{sec:brodal}

In this section we describe our algorithm and analysis for Brodal queue~\cite{Brodal96} selection. This section will assume familiarity with Brodal queues, but we give an overview of needed concepts for readers who are not already familiar in Appendix~\ref{sec:boverview}. For more understanding, we suggest reading the original Brodal queue paper~\cite{Brodal96}.

Brodal queue~\cite{Brodal96} selection will differ from Fibonacci heap selection in a couple ways. The most important is handling violating nodes; that is, the potentially $O(n)$ nodes that do not satisfy heap-order. We treat violating nodes with on-the-fly conversion into a proper heap-ordered tree. We apply invariants O4 and O5 to show there cannot be too many violating nodes encountered throughout the selection.
The second is that the rigid rank structure of Brodal queues can actually allow us to simplify the proof of degree bounds compared to the approach taken in the previous section. Namely, rather than lower bounding the sizes of unselected subtrees, we can directly upper bound the number of nodes of high degree. We first give a lemma from~\cite{Brodal96}\footnote{While this is not an explicit lemma in~\cite{Brodal96}, it is explicitly stated and proven on page 2.}.

\begin{lemma}[Brodal~\cite{Brodal96}]
	\label{lem:brodal-node}
	A node $x$ with rank $r(x)$ in a Brodal queue has subtree of size at least $2^{r(x)+1}-1$.
\end{lemma}

Lemma~\ref{lem:brodal-node} implies the maximum rank $r$ in a Brodal queue is at most $\log_2(n)$. 

We can use S2 and Lemma~\ref{lem:brodal-node} to get a bound on the number of nodes in a Brodal queue of a particular rank.

\begin{lemma}
	\label{lem:brodal-ranks}
	A Brodal queue has at most one node of rank $\lceil\log_2(n)\rceil$, two nodes of rank $\lceil\log_2(n)\rceil-1$, and in general at most $2^i$ nodes of rank $\lceil\log_2(n)\rceil-i$, where $i \leq \lceil\log_2(n)\rceil$.
\end{lemma}
\begin{proof}
	By S2, nodes of the same rank cannot be descendants of each other. This implies their subtrees are disjoint. By Lemma~\ref{lem:brodal-node}, a node of rank $\lceil\log_2(n)\rceil-i$ has at least $n/2^{i-1}-1 \geq n/2^i$ descendants. Thus there can only be $2^i$ such nodes.
\end{proof}

\begin{theorem}
\label{thm:brodselect}
\sloppy
	Brodal queues support \texttt{SelectK($k$)}, \texttt{ExtractK($k$)}, and \texttt{Delete($e_1, \ldots, e_k$)} in $O(k \log (n/k))$ worst-case time.
\end{theorem}
\begin{proof}
	As in the Brodal queue delete minimum operation, we first empty tree $T_2$ by moving all children of $t_2$ to $T_1$ and making $t_2$ a rank $0$ child of $T_1$. We then call \texttt{Soft-Select-Heapify($t_1$)} to select the $k$ smallest nodes in $T_1$.
	
	Whenever we reach a node $y$, we consider nodes of $V(y)$ and $W(y)$ to be children of $y$ in the selection algorithm. From invariant O2, we know that $x\geq y$. It is thus possible to reach a node $x$ via its proper parent or a node $y$ in which $x \in V(y) \cup W(y)$. Upon encountering $x$, we check to make sure it was not already selected. Since $x$ is in the violating set $V(y)$ or $W(y)$ of at most one node $y$, the total extra work for encountering $x$ on two paths is $O(1)$, totaling $O(k)$ for all $k$ selected nodes. Algorithm \texttt{Soft-Select-Heapify()} can process the selection by considering the heap to be constructed in this way since the tree $T$ we construct on the fly still remains heap-ordered.
	
	We now consider the cost of the selection. We must show $D(T, k) = O(k \log(n/k))$, where $T$ is the heap-ordered tree we construct on the fly. Observe that the degree of a node $x$ in $T$ is the sum of the number of violating nodes in its sets $V(x)$ and $W(x)$ and the number of its proper children. Observe that the number of proper children of a node $x$ is bounded by $7r(x)$. Thus to bound the number of proper children of selected nodes, it suffices to find an upper bound on the total sum of ranks of selected nodes.
	
	Let the ranks of the $k$ selected nodes in non-increasing order be $r_1, \ldots, r_k$. Then by Lemma~\ref{lem:brodal-ranks}, $r_1 \leq \lceil\log_2(n)\rceil$, $r_2, r_3 \leq \lceil\log_2(n)\rceil-1$, $r_4, \ldots, r_7 \leq \lceil\log_2(n)\rceil-2$, and $r_k \leq \lceil\log_2(n)\rceil - \lfloor\log_2(k)\rfloor \leq \log_2(n/k)+2$. The total sum of ranks is thus at most
	\[
	2k+\sum_{i=0}^\infty \frac{k}{2^i} (\log_2(n/k) + i) = k\log_2(n/k) + 4k.
	\]
	It follows the total sum of proper children of selected nodes is $O(k\log(n/k))$.
	
	We can apply a similar strategy to bound the total number of nodes in violating sets $V$ and $W$ of selected nodes. Divide the nodes into two categories: those with rank greater than or equal to $\log_2(n)-\log_2(k)$ and those with rank less than $\log_2(n)-\log_2(k)$. By Lemma~\ref{lem:brodal-ranks} again and by employing the above argument, there are at most $O(k)$ nodes in the first category.
	
	We can bound the number of nodes in the second category with invariants O4 and O5. Invariant O4 states that the $W$ sets of selected nodes contain at most $6$ nodes per rank. Invariant O5 states that $V$ sets of selected nodes contain at most $m\alpha$ nodes of rank less than $m$. Together they imply a bound of $(6+\alpha)\log_2(n/k)$ nodes in total in the second category. As $\alpha = O(1)$, this implies the total number of nodes in violating sets of selected nodes is $O(k\log(n/k))$. Since the number of proper children of selected nodes is also $O(k\log(n/k))$, we have $D(T,k) = O(k \log(n/k))$. By Theorem~\ref{thm:selection}, this implies the heap-ordered tree selection takes $O(k \log(n/k))$ worst-case time.
	
	This shows \texttt{SelectK($k$)} takes $O(k \log(n/k))$ worst-case time on a Brodal queue. To prove $O(k \log(n/k))$ worst-case runtime for \texttt{ExtractK($k$)}, we must rebuild the Brodal queue with the $k$ smallest nodes removed. We can do so as follows. We remove nodes other than $t_1$ one-by-one. Consider the process for a node $x$. First, we remove $x$. This may cause $p(x)$ to have only one child of rank $r(x)$, violating S4. If $r(p(x)) > r(x)+1$, we can remove the other rank $r(x)$ node from $p(x)$ and make it a child of $t_1$. Otherwise, we can find a rank $r(x)$ child of the root and replace $x$ with it. This may create a violation; if so, we can add it to $W(t_1)$. We then add the proper children of $x$ below $t_1$. Further, we must deal with $V(x)$ and $W(x)$. We can simply add them to $W(t_1)$.
	
	We can bound the total cost of adding children of removed nodes below $t_1$ as the above analysis gives for the bound on the number of proper children, at $O(k \log(n/k))$. Similarly, we can bound the number of violations created at $O(k)$ and the number of violations added to $W(t_1)$ at $O(k \log(n/k))$, again following the above analysis. At this point we can then remove $t_1$, following the extract minimum procedure stated in \texttt{DeleteMin($Q$)} of~\cite{Brodal96}. The total time taken to remove the $k$ smallest elements is $O(k \log(n/k))$.
	
	As in the proof of Theorem~\ref{thm:fibselect}, we again do not need the property that the $k$ nodes removed in \texttt{ExtractK($k$)} are minimal. We can thus support \texttt{Delete($e_1, \ldots, e_k$)} as follows. We again first empty tree $T_2$ by moving all children of $t_2$ to $T_1$ and making $t_2$ a rank $0$ child of $T_1$. If any $e_i = t_1$, we remove each node one-by-one as in \texttt{ExtractK($k$)}, but save $t_1$ for the final removal as in \texttt{DeleteMin($Q$)} of~\cite{Brodal96}. Otherwise, we remove each node $e_i$ one-by-one exactly as in the above procedure for \texttt{ExtractK($k$)}, skipping the final \texttt{DeleteMin($Q$)} procedure of~\cite{Brodal96}. Either way the above analysis indicates the total cost will be $O(k \log(n/k))$ worst-case time.
\end{proof}

\section{Optimal lazy search trees}
\label{sec:lst}

Lazy search trees~\cite{Sandlund20} are comparison-based data structures that support the following operations on a dynamic set $S$ with $|S| = n$. (It is straightforward to extend data structures to multisets.) Lazy search trees are designed for scenarios where the number of insertions is larger than the number of queries. The element of rank $r$ is the $r$th smallest element in the set $S$. The operations of lazy search trees are referred to as a \emph{sorted dictionary}.

Lazy search trees support the following operations that change the set $S$.
\begin{itemize}
	\item \texttt{Construction($S$)} $\ce$ Construct a sorted dictionary on the unsorted set $S$.
	\item \texttt{Insert($e$)} $\ce$ Add element $e$ to $S$; 
	(this increments $n$).
	\item \texttt{Delete($e$)} $\ce$ Delete $e$ from $S$, with a pointer to $e$; (this decrements $n$).
	\item \texttt{ChangeKey($e$,\,$v$)} $\ce$ Change the key of the element $e$ (with pointer to it) to $v$.
	\item \twopart{\texttt{Split($r$)} $\ce$ }{Split $S$ at rank $r$, returning two sorted dictionaries $T_1$ and $T_2$ of $r$ and $n-r$ elements, respectively, such that for all $x \in T_1$, $y \in T_2$, $x \leq y$.}
	\item \twopart{\texttt{Merge($T_1$,\,$T_2$)} $\ce$ }{Merge sorted dictionaries $T_1$ and $T_2$ and return the result, given that for all $x \in T_1$, $y \in T_2$, $x \leq y$.}
\end{itemize}

Lazy search trees support rank-based queries to get information about the set. Informally, a rank-based query is a query computable in $O(\log n)$ time on a (possibly augmented) binary search tree and in $O(n)$ time on an unsorted array. The permitted queries include: 

\begin{itemize}
    \item \texttt{Rank($k$)} $\ce$ Return the rank of key $k$ in the set $S$.
    \item \texttt{Select($r$)} $\ce$ Return the element of rank $r$ in $S$. 
    \item \texttt{Contains($k$)} $\ce$ Return \texttt{true} if there exists an element $e \in S$ with key $k$; otherwise return \texttt{false}. 
    \item \texttt{Successor($k$)} $\ce$ Return the successor of the element $e$ in set $S$.
    \item \texttt{Predecessor($k$)} $\ce$ Return the predecessor of the element $e$ in set $S$.
    \item \texttt{Minimum()} $\ce$ Return the minimum element of set $S$
    \item \texttt{Maximum()} $\ce$ Return the maximum element of set $S$.
\end{itemize}

For a formal description of permissible queries, see the original lazy search trees paper, Section 4~\cite{Sandlund20}.

Lazy search trees seek to improve the $O(\log n)$ time per-operation complexity given by binary search trees as a sorted dictionary by employing a fine-grained complexity analysis, not unlike that done in dynamic optimality literature~\cite{Sleator85,Wilber89,Cole2000a,Cole2000b,Iacono01,Demaine07,Demaine09}.

Instead of sorting elements upon insertion, sorting is delayed until query operations. Elements are stored in a partition into gaps $\Delta_1, \ldots, \Delta_m$ such that for $x \in \Delta_i$ and $y \in \Delta_{i+1}$, $x \leq y$. Inserted elements are placed into a gap respecting the key-order partition. Upon query, the gap $\Delta_i$ containing rank $r$ is split into two gaps $\Delta'_i$ and $\Delta'_{i+1}$ such that $|\Delta'_i| + \sum_{j=1}^{i-1} |\Delta_j| = r$ and for $x \in \Delta'_i$, $y \in \Delta'_{i+1}$, $x \leq y$.

\sloppy
The results of~\cite{Sandlund20} are \texttt{Construction($S$)} in $O(n)$ time where $|S| = n$, \texttt{Insert()} into a gap $\Delta_i$ in $O(\min(\log(n/|\Delta_i|) + \log \log |\Delta_i|,\: \log q))$ worst-case time ($q$ is the number of queries), \texttt{RankBasedQuery()} that splits a gap $\Delta_i$ into a smaller gap of size $k$ and a larger gap of size $|\Delta_i|-k$ in $O(k \log(|\Delta_i|/k) + \log n)$ amortized time, \texttt{Delete()} in $O(\log n)$ worst-case time, \texttt{Split($r$)} in time as in \texttt{RankBasedQuery($r$)}, and \texttt{Merge()} in $O(\log n)$ worst-case time.

The performance of \texttt{ChangeKey($e$,\,$v$)} is not as easily stated. In general it can be supported in $O(\log n)$ worst-case time as in deletion and reinsertion. More efficient runtimes are possible dependent on the rank of the element $e$ and the nature of the gap $\Delta_i$ to which $e$ belongs (and remains after the key-change).

Each gap is either \emph{zero-sided}, \emph{left-sided}, \emph{right-sided}, or \emph{two-sided}. If no queries have occurred, the single gap $\Delta_1$ is zero-sided; if queries have occurred only on the left boundary of $\Delta_i$ it is left-sided, only on the right boundary it is right-sided, and otherwise (the typical case) it is two-sided. In this way, all but the leftmost and rightmost gaps are two-sided\footnote{The distinction between zero-sided, left-sided, and right-sided gaps is motivated by the applications of lazy search trees; namely, to provide a single data structure that can be used both as an efficient priority queue and binary search tree replacement. It is not a restriction based on the inner workings of the data structure, and indeed many variations of the data structure are possible.}. A zero-sided gap $\Delta_i$ supports any key-change operation within $\Delta_i$ in $O(1)$ time. A left-sided gap $\Delta_i$ supports decrease-key within $\Delta_i$ in $O(\min(\log q, \log \log |\Delta_i|))$ time (recall $q$ is the number of queries). A right-sided gap $\Delta_i$ supports increase-key within $\Delta_i$ in the same time complexity. A two-sided gap $\Delta_i$ supports decrease-key on elements less than or equal to the median of $\Delta_i$ and increase-key on elements larger than or equal to the median of $\Delta_i$, again in the same time complexity.

\fussy
As previously stated, if we take $B = \sum_{i=1}^m |\Delta_i| \log_2(n/|\Delta_i|)$, lazy search trees serve an operation sequence of $n$ insertions and $q$ distinct queries in $O(B + \min(n \log \log n, n \log q))$ time, where the $\Delta_i$ referred to in the bound is the resulting $\Delta_i$ after completion of the operation sequence. Further, $\Omega(B+n)$ is a lower bound~\cite{Sandlund20}. By using selectable heaps in place of the interval data structure in~\cite{Sandlund20}, we can achieve three new results for lazy search trees.

\begin{theorem}
	\label{thm:lst}
	\begin{enumerate}
		\item Lazy search trees can support insertion into gap $\Delta_i$ in $O(\log (n/|\Delta_i|))$ worst-case time and change-key in $O(1)$ worst-case instead of $O(\min(\log q, \log \log |\Delta_i|))$ worst-case time in the conditions stated above, while matching previous time bounds for all other operations. Taking $B = \sum_{i=1}^m |\Delta_i| \log_2(n/|\Delta_i|)$, lazy search trees serve a sequence of $n$ insertions and $q$ distinct queries in $O(B+n)$ time, which is optimal.
		\item Split of a two-sided gap $\Delta_i$ can be done in worst-case instead of amortized time. Split of a left-sided or right-sided gap $\Delta_i$ can be done in worst-case instead of amortized time if the larger resulting gap is left-sided or right-sided, respectively.
		\item Merge of two left-sided gaps or two right-sided gaps can be performed in $O(1)$ amortized or worst-case time. Alternatively, change-key can be made $O(\log n)$ time, the merge of any two gaps can be supported in $O(1)$ time, and all operations of the lazy search tree can be made worst-case.
	\end{enumerate}
	
\end{theorem}

\begin{proof}
    We create a new interval data structure\footnote{The term ``interval data structure" was used in~\cite{Sandlund20} due to the representation of elements within a gap into a second-level key-order partition into intervals, analogous to the gaps on the first level. The ``interval data structure" discussed herein is based on selectable heaps, so at this point the name is a misnomer, but we maintain the nomenclature for consistency.} based on selectable heaps. If the gap is zero-sided, the interval data structure is an unsorted array or linked list. If the gap is left-sided, the interval data structure is a selectable min-heap. If the gap is right-sided, the interval data structure is a selectable max-heap. Otherwise, the gap is two-sided, and we partition the elements roughly into thirds. The smallest third of the elements we make into a min-heap, the largest into a max-heap, and the middle third we keep unsorted. We maintain that each third contains at least a $1/3-\epsilon$ for $0 < \epsilon < 1/6$ fraction of the total elements in the gap. We can maintain static separator elements between the thirds to ensure a valid partition.
	
	As in~\cite{Sandlund20}, insertion first locates the gap $\Delta_i$ in which the inserted element belongs in $O(\log(n/|\Delta_i|))$ worst-case time via a biased search tree~\cite{Bent85}. Insertion within the interval data structure then considers which third to place the element, if applicable, then does so in $O(1)$ worst-case time if the selectable heap is Fibonacci~\cite{Fredman87} or Brodal~\cite{Brodal96} (or if it is an unsorted array). Change-key is supported as decrease-key in a left-sided heap or increase-key in a right-sided heap, in $O(1)$ time. For a two-sided gap, any element in the middle third can have its key increased or decreased via removal and re-insertion into the proper third in $O(1)$ time, otherwise the min-heap supports decrease-key and the max-heap increase-key, ultimately satisfying the necessary constraints for efficient change-key, in $O(1)$ time. If a Brodal queue is used, the complexity is worst-case.
	
	Query works as follows. As in~\cite{Sandlund20}, first the gap $\Delta_i$ in which $r \in \Delta_i$ is found in $O(\log n)$ time. If $\Delta_i$ is zero-sided we amortize the work against the total number of elements in zero-sided gaps, as in~\cite{Sandlund20}, performing the operation in $O(1)$ amortized time. If $\Delta_i$ is two-sided, then it should be possible to determine which third the query rank $r$ falls into, if applicable, in $O(1)$ time. If in the middle third, we answer the query in $O(|\Delta_i|)$ time, rebuilding $\Delta_i$ into $\Delta'_i$ and $\Delta'_{i+1}$ as previously described. Otherwise, we must perform selection in a min- or max-heap; without loss of generality, assume it is a min-heap. We repeat the following. We select the smallest $2^j$ elements into a set $X$ and the smallest $2^{j+1}$ elements into a set $Y$, starting at $j = 0$. By Definition 6 in the original paper~\cite{Sandlund20}, we can determine which of $X$ or $Y$ contains $r$. If it is $Y$, we continue with $j \leftarrow j+1$; otherwise, we stop and answer the query. We break $\Delta_i$ into $\Delta'_i$ and $\Delta'_{i+1}$ so that $|\Delta'_i| + \sum_{j=1}^{i-1} |\Delta_i| = r$, by extracting $k = r - \sum_{j=1}^{i-1} |\Delta_i|$ elements from the heap. We make $\Delta'_i$ into a new two-sided gap in $O(k)$ time. The existing structure of gap $\Delta_i$ becomes $\Delta'_{i+1}$.

	\sloppy
	The time complexity of query can be proven as follows. The time taken for the selections, by Theorems~\ref{thm:fibselect} and~\ref{thm:brodselect}, is proportional to no more than $\sum_{i = 0}^\infty k/2^i \log (n2^i/k) = O(k \log (n/k))$. Extraction similarly takes $O(k \log (n/k))$ time. If a Brodal queue is used, the time bound is worst-case. However, $k$ may be larger than $|\Delta_i|/2$ if $\Delta_i$ was a left-sided or right-sided gap. In this case we can amortized against the total number of elements in one-sided gaps, as is done in~\cite{Sandlund20}, to perform the operation in $O(1)$ amortized time.
	
	\fussy
	Operation \texttt{Construction(S)} can be completed via insert. Operation \texttt{Delete($e$)} can be performed as priority queue deletion, in $O(\log n)$ time. Operation \texttt{Split($r$)} is performed as query and then an operation on the biased search tree gap data structure, as in~\cite{Sandlund20}. Finally, \texttt{Merge($T_1$,\,$T_2$)} similarly occurs at the gap level.

	One detail remains, which is the maintenance of the partition into thirds in a two-sided gap into fractions of size at least $1/3-\epsilon$ of the total gap size. For every $k$ elements inserted, removed, or key-changed in an operation, we perform $O(k)$ work towards building a new copy of the interval data structure with a more-accurate partition, as described in~\cite{overmars1983design}. After the new version is constructed it is caught up to the operations performed during construction at twice the pace they occur. When it is caught up the current data structure is thrown away and construction on a new data structure begins. This allows worst-case time maintenance of the partition while keeping all operation complexity the same.
	
	The above shows parts 1 and 2 of Theorem~\ref{thm:lst}. We now consider part 3. Brodal and Fibonacci heaps support $O(1)$ time merge, in worst-case and amortized time, respectively. As left-sided and right-sided gaps are just selectable heaps, we can simply merge the heaps. For the alternative approach, we forget about left-sided, right-sided, or two-sided gaps and store all elements of a gap in both a min- and max-heap, with each element containing a pointer to the other in the opposite heap. All operations work as above, except now for query, when we extract elements from the min- or max-heap, we also delete the elements in the other heap. As stated in Theorems~\ref{thm:fibselect} and~\ref{thm:brodselect}, deletion of the $k$ elements can be completed in $O(k \log(n/k))$ amortized or worst-case time, respectively. Now, since any two gaps has both a min- and a max-heap, we can support the merge of two arbitrary gaps in $O(1)$ time by simply merging the two min- and max-heaps. Change-key is supported only by deletion and re-insertion, in $O(\log n)$ time. All operations can be made worst-case via use of a Brodal queue.
\end{proof}

Theorem~\ref{thm:lst} answers open problems 2, 4, and 5 from~\cite{Sandlund20}, though we give worst-case performance only in the general case of two-sided gaps. It appears this may be necessary if change-key is to be supported as a decrease-key operation for all elements in a left-sided gap (analogously, increase-key in a right-sided gap), which is what provides optimal performance as a priority queue. Specifically, the performance of change-key and the ability to perform quick splitting of a gap are coupled. If a lazy search tree is used as a min-heap, allowing decrease-key of all elements, then we cannot find the maximum element in $O(\log n)$ worst-case time. We need to rebuild the data structure so that ranks close to the maximum and minimum can be found efficiently, which necessarily requires the change-key operation to be decrease-key on elements with rank closer to the minimum and increase-key on elements with rank closer to the maximum.

Finally, we observe that the alternative insertion complexity of $O(\log q)$ is also achieved in the version of lazy search trees stated herein. The number of elements in the gap data structure is bounded by $q$~\cite{Sandlund20}, and insertion into the interval data structure based on selectable heaps takes $O(1)$ time. However, $O(\log q)$ vs. $O(\log (n/|\Delta_i|))$ time insertion does not impact overall time complexity on any sequence of operations, so we leave it out of the statement of Theorem~\ref{thm:lst}.

\section{Conclusion}
\label{sec:conclude}

In this paper we have shown that the $O(\log n)$ time extract-minimum function of efficient priority queues can be generalized to the extraction of the $k$ smallest elements in $O(k \log(n/k))$ time. We further show selection of the $k$ smallest elements without extraction and deletion of any $k$ elements can be performed also in $O(k \log(n/k))$ time.

We apply selectable heaps to lazy search trees~\cite{Sandlund20}, giving an optimal data structure in the gap model, adding a merge function when used as a priority queue, and providing worst-case runtimes in the general case of two-sided gaps. Any further theoretical improvement in lazy search trees would require abstraction to an even more fine-grained model.

We believe by using selectable heaps, we can achieve a more straightforward approach to lazy search trees than that of~\cite{Sandlund20}. With an understanding of Fibonacci heaps~\cite{Fredman87} or Brodal queues~\cite{Brodal96}, the technical arguments required herein are slightly less involved. However, the approach of~\cite{Sandlund20} based on first principles has its own merit. In~\cite{Sandlund20}, the data structure satisfies an $O(\min(n, q))$ pointer bound, where $n$ is the number of elements and $q$ the number of queries. Further, the simple priority queue developed which supports extraction natively is surely more practical. In~\cite{Sandlund20}, it is shown in the experiments that a rudimentary implementation of the described data structure when the number of insertions approaches $n\geq 1\,000\,000$ can outperform binary search trees in the following two ways: when the number of queries $q$ is significantly smaller than the number of insertions $n$ (like $q \leq \sqrt{n}$), or when the queries are highly non-uniform, like large range queries or a priority queue operation sequence. In contrast, considering the complicated structures of Fibonacci heaps~\cite{Fredman87} or Brodal queues~\cite{Brodal96}, the lazy search tree improvement discussed herein are likely to mostly be of theoretical interest.

For future work, it would be interesting to see if selectable heaps have further applications outside of the straightforward transitive applications through lazy search trees, such as an optimal online multiple selection algorithm~\cite{Dobkin81}. It would also be interesting to see if selection can be supported on a priority queue with optimal worst-case guarantees in the pointer machine model, such as strict Fibonacci heaps~\cite{Brodal12}. Such a data structure would allow a lazy search tree with worst-case guarantees on the pointer machine, whereas the version discussed here requires internal use of arrays. Finally, it may be possible to support \texttt{SelectK($k$)} in $O(k)$ time, as do lazy search tree priority queues, but while retaining optimal time bounds for the remaining operations.

\bibliographystyle{alpha}
\bibliography{selectable-heaps}

\appendix

\section{Brodal queue overview}
\label{sec:boverview}

Brodal queues~\cite{Brodal96} are much more complicated than Fibonacci heaps, but allow all operations in worst-case time. Nodes are stored in a tree $T_1$ or possibly another tree $T_2$ which is incrementally merged into $T_1$. Each node has a non-negative integer rank assigned to it. For each node $x$, we use $p(x)$ as the parent of $x$, $r(x)$ as the rank\footnote{Rank is used here as an assigned parameter. The term rank is used for different concepts in Section~\ref{sec:intro}.} of $x$, and $n_i(x)$ as the number of children of $x$ of rank $i$. Finally, we use $t_i$ as the root of $T_i$. Nodes which satisfy heap order, i.e. they are larger than their parents, are called good nodes. Nodes which are not good are called \textit{violating} nodes.

Brodal queues use 13 invariants to maintain the structure. The invariants can be classified into three sets S, O and R respectively. The invariants in the set S apply to all nodes in the Brodal queue, while the invariants in the set R apply to the root of the trees. The invariants in set O control the violations.

For any node $x$, the following invariants are satisfied.
\begin{itemize}[label={}]
    \item S1 : If $x$ is a leaf, then $r(x) = 0$,
    \item S2 : $r(x) < r(p(x))$,
    \item S3 : if $r(x) > 0$, then $n_{r(x)-1}(x) \geq 2$,
    \item S4 : $n_i(x) \in \{0, 2, 3, \dots, 7\}$,
    \item S5 : $T_2 = \emptyset $ or $r(t_1) \leq r(t_2)$.
\end{itemize}

 Each node $x$ has rank $r(x)$ and at most $7$ children of any rank less than $r(x)$ (Invariants S2 and S4). Additionally, $x$ cannot have a single child of any rank (S4) and must have at least two children of rank $r(x)-1$ (S3). Leaves have rank $0$ (S1).

Beside the good nodes, Brodal queues also allow violating nodes. To keep track
of the violating nodes, each node $x$ is associated with two subsets $V(x)$ and $W(x)$. If a node $y$ is smaller than its parent $p(y)$, the violation is stored in a violating set $V(x)$ or $W(x)$ for some node $x \leq y$. Despite the fact that each node $x$ might need to maintain non-empty subsets $V(x)$ and $W(x)$, elements will only be inserted to either $V(t_1)$ or $W(t_1)$ ($t_1$ is the root of $T_1$).
The $V$ sets take care of large violations, i.e. violations that have rank larger than $r(t_1)$ when they are created will be added to $V(t_1)$. The $W$ sets handle smaller violations, i.e. violations that have rank less than or equal to $r(t_1)$. We use $w_i(x)$ to denote the number of nodes of rank $i$ in the set $W(x)$.
Brodal uses the constant $\alpha$ for the number of large violations that can be created between two increases in the rank of $t_1$.
In the original paper~\cite{Brodal96}, Brodal uses the following set of invariants to maintain the violations. 

\begin{itemize}[label={}]
    \item O1 : $t_1$ = min $T_1 \cup T_2$,
    \item O2 : if $y \in V(x) \cup W(x)$, then $y \geq x$,
    \item O3 : if $y < p(y)$, then an $x \neq y$ exists such that $y \in V (x) \cup W(x)$,
    \item O4 : $w_i(x) \leq 6$,
    \item O5 : if $V (x) = (y_{|V (x)|} , \dots , y_2, y_1)$, then $r(y_i) \geq \left \lfloor{(i-1)/\alpha}\right \rfloor $  for $i = 1, 2, \dots , |V (x)|$, where $\alpha$ is a constant.
\end{itemize}

From the above invariants, we can obtain the upper bound of the size of $V$ and $W$ sets. The set $W(x)$ contains at most $6$ violations of any rank (O4) and the set $V(x)$ contains at most $m\alpha$ violations of rank less than $m$, where $\alpha$ is a constant (O5).

At the roots $t_1$ of $T_1$ and $t_2$ of $T_2$, invariants are stronger. Each root $t_j$ has at least two children of every rank (R1). 

\begin{itemize}[label={}]
    \item R1 : $n_i(t_j ) \in \{2, 3, \dots , 7\}$ for $i = 0, 1, \dots , r(t_j)-1$,
    \item R2 : $|V (t_1)| \leq \alpha r(t_1)$,
    \item R3 : if $y \in W(t_1)$, then $r(y) < r(t_1)$.
\end{itemize}
   
   A \emph{guide} data structure ~\cite{Kaplan99,Elmasry07} is used to efficiently manage the invariants R1 and O4. Given a sequence of integer 
    variables $x_k, x_{k-1}, \dots , x_1$ and $x_i \leq T$ for some threshold $T$, we can only perform \texttt{Reduce(i)} operations on the sequence which decrease $x_i$ by at least two and increase $x_{i+1}$ by at most one. The $x_i$s can be forced to increase and decrease by one, but for each change in an $x_i$ we are allowed to do $O(1)$ \texttt{Reduce} operations to prevent any
    $x_i$ from exceeding $T$. The guide data structure will tell us which
    operations to perform in $O(1)$ time.

    Notice that, if we have at least three nodes of the same rank $i$, we can combine those nodes and create a new node with rank $i+1$. Similarly, we can also split a node of rank $i+1$ to several nodes of rank at most $i$. We call the first operation \textit{linking} and the latter \textit{delinking}.
    For each node, if we count the number of children of each rank and list the number from highest rank to lowest, we obtain a sequence of integer 
    variables $x_k, x_{k-1}, \dots , x_1$. 
    We can perform linking and delinking operation for the children of the same node, and we can reflect the changes of the rank in the above sequence by performing \texttt{Reduce(i)} operations. 
    We can use two instances of the guide data structure for each node to maintain both a lower and upper bound on the number of children of each rank so that addition and removal of a child of any rank can be supported at the roots $t_1$ and $t_2$ in $O(1)$ worst-case time.
    A guide data structure is also used to maintain the upper bound given in Invariant O4 on $W(t_1)$. Two violations of the same rank $r$ can be reduced to at most one of rank $r+1$ in worst-case $O(1)$ time.

\end{document}